\documentclass[letterpaper,journal]{IEEEtran} 
\usepackage{amsmath,amsfonts}
\usepackage{array}
\usepackage[caption=false,font=normalsize,labelfont=sf,textfont=sf]{subfig}
\usepackage{amssymb}
\usepackage{longtable}
\usepackage{textcomp}
\usepackage{stfloats}
\usepackage{url}
\usepackage{verbatim}
\usepackage{graphicx}
\usepackage{xcolor}
\usepackage{fancyhdr}

\definecolor{customColor}{HTML}{FBB41A}

\usepackage{multirow}
\usepackage{booktabs}
\usepackage{cite}
\usepackage{hyperref}

\usepackage{cleveref}

\usepackage{amsthm}
\usepackage[ruled,vlined]{algorithm2e}

\newtheorem{theorem}{Theorem}

\usepackage{algorithmic}

\begin{document}

\title{Q-Tag: Watermarking Quantum Circuit Generative Models}

\author{
Yang Yang, Yuzhu Long, Han Fang, Zhaoyun Chen, Zhonghui Li, Weiming Zhang, Guoping Guo
\thanks{
This work was supported in part by the National Natural Science Foundation of China under Grant 62272003, in part by the Quantum Science and Technology-National Science and Technology Major Project under Grant 2021ZD0302300, in part by the Science and Technology Major Project of Anhui Province under Grant 202423s06050001 and in part by the National Key Research and Development Program of China (Grant Nos. 2023YFB4502500 and 2024YFB4504100)(Corresponding author: Han Fang and Weiming Zhang.)

Yang Yang is with the School of Electronic and Information Engineering, Anhui University, Hefei, Anhui 230601, China, and the Institute of Artificial Intelligence, Hefei Comprehensive National Science Center, Hefei, Anhui, 230088, China (e-mail: sky\_yang@ahu.edu.cn).

Yuzhu Long is with the School of Electronic and Information Engineering, Anhui University, Hefei, Anhui 230601, China (e-mail: p23301192@stu.ahu.edu.cn).  
  
Han Fang is with the School of Computing, National University of Singapore, Singapore (e-mail: fanghan@nus.edu.sg)

Zhaoyun Chen is with the Institute of Artificial Intelligence, Hefei Comprehensive National Science Center, Hefei, Anhui, 230088, China.

%Kejiang Chen is with the School of Cyber Science and Technology, University of Science and Technology of China, Hefei, Anhui 230026, China.

Zhonghui Li is with the School of Electronic and Information Engineering, Anhui University, Hefei, Anhui 230601, China

Weiming Zhang is with the School of Cyber Science and Technology, University of Science and Technology of China, Hefei, Anhui 230026, China (e-mail: zhangwm@ustc.edu.cn).

Guoping Guo is with the Laboratory of Quantum Information, University of Science and Technology of China, Hefei, Anhui, 230026, China, and the Origin Quantum Computing Technology Company, Hefei, Anhui, 230026, China.

}
}

% \markboth{SUBMITTED TO IEEE TRANSACTIONS ON SERVICES COMPUTING}%
% {YY \MakeLowercase{\textit{et al.}}: Q-Tag: Watermarking Quantum Circuit Generative Models}

\maketitle

\thispagestyle{fancy}
\fancyhead[L]{\scriptsize \leftmark}  
\fancyhead[R]{\thepage} 
\lfoot{}

% \cfoot{\small Copyright © 2025 IEEE. Personal use of this material is permitted. However, permission to use this material for any other purposes must be obtained from the IEEE by sending an email to pubs-permissions@ieee.org.}

\renewcommand{\headrulewidth}{0mm}

\begin{abstract}
Quantum cloud platforms have become the most widely adopted and mainstream approach for accessing quantum computing resources, due to the scarcity and operational complexity of quantum hardware. In this service-oriented paradigm, quantum circuits, which constitute high-value intellectual property, are exposed to risks of unauthorized access, reuse, and misuse. Digital watermarking has been explored as a promising mechanism for protecting quantum circuits by embedding ownership information for tracing and verification. However, driven by recent advances in generative artificial intelligence, the paradigm of quantum circuit design is shifting from individually and manually constructed circuits to automated synthesis based on quantum circuit generative models (QCGMs). In such generative settings, protecting only individual output circuits is insufficient, and existing post hoc, circuit-centric watermarking methods are not designed to integrate with the generative process, often failing to simultaneously ensure stealthiness, functional correctness, and robustness at scale. These limitations highlight the need for a new watermarking paradigm that is natively integrated with quantum circuit generative models. In this work, we present \textbf{\textit{the first watermarking framework for QCGMs}}, which embeds ownership signals into the generation process while preserving circuit fidelity. We introduce a symmetric sampling strategy that aligns watermark encoding with the model’s Gaussian prior, and a synchronization mechanism that counteracts adversarial watermark attack through latent drift correction. Empirical results confirm that our method achieves high-fidelity circuit generation and robust watermark detection across a range of perturbations, paving the way for scalable, secure copyright protection in AI-powered quantum design.

\end{abstract}

\begin{IEEEkeywords}
Quantum Circuit, IP Protection, Latent Diffusion Model, Quantum Circuit Generative Models, Watermarking.
\end{IEEEkeywords}

\section{Introduction}\label{Introduction}
\IEEEPARstart{Q}{uantum} machine learning (QML)~\cite{QML,quantummachinelearning} has emerged as a promising frontier in quantum computing, offering the potential to accelerate tasks in classification, regression, generative modeling, and beyond by exploiting quantum parallelism~\cite{TextClassification,GenerativeLearning,AlphaTensor_2024,genQC}. These capabilities are poised to revolutionize fields ranging from chemistry and cryptography to machine learning~\cite{lu2024data,fesquet2024demonstration,jerbi2024shadows}. As quantum algorithms grow increasingly complex and application-driven, quantum circuits, the executable representations of quantum algorithms, play a central role in bridging theoretical models and real hardware implementation~\cite{informationscramblingqc}. In essence, quantum circuits form the concrete runtime layer of QML, translating high-level abstractions into physical gate operations. Given their function-specific structure and optimization, these circuits are not generic but encode significant intellectual property (IP)~\cite{IP-AI,IP-DEEPW,IP-DEEPIPR}, often the result of proprietary design or computationally expensive synthesis.

Quantum cloud platforms have become the most widely adopted and mainstream approach for accessing quantum computing resources, largely due to the scarcity and operational complexity of quantum hardware. Leading platforms such as IBM Quantum~\cite{IBM}, Amazon Braket~\cite{Amazon}, and Origin Quantum~\cite{OriginQuantum} provide standardized cloud-based interfaces that allow users to remotely submit and execute quantum circuits without maintaining physical quantum devices, thereby significantly lowering the barrier to quantum computing and enabling scalable access to quantum resources~\cite{Majumder2025QuanFraud}. However, outsourcing circuit execution to quantum cloud platforms also introduces non-negligible intellectual property (IP) risks. Once uploaded to the cloud, quantum circuits may be stored, inspected, replicated, or reused by untrusted third parties without the designer’s consent or awareness. To mitigate these risks, several studies have explored post hoc quantum circuit watermarking techniques that embed ownership information directly into quantum circuits via appending a rotation gate on ancilla qubits combined with other gates at the circuit output, or inserting a pair of random gates and their inverses in the circuit interior separated by barriers~\cite{Roy2025Watermarking}, enabling copyright verification and forensic tracing after deployment. These approaches represent early attempts to protect circuit-level intellectual property in quantum cloud environments.

Driven by recent advances in generative artificial intelligence, the paradigm of quantum circuit design is undergoing a fundamental shift from individually and manually constructed circuits to automated synthesis based on quantum circuit generative models (QCGMs)~\cite{genQC}. By leveraging data-driven and learning-based approaches, QCGMs enable efficient and scalable generation of quantum circuits conditioned on task objectives or optimization criteria, making them particularly suitable for deployment as cloud-based quantum design services. This shift in design paradigm fundamentally changes the requirements for copyright protection. In generative settings, quantum circuits are no longer isolated design artifacts but outputs sampled from a learned generative process, rendering the protection of individual circuit instances insufficient. Instead, copyright protection must be tightly coupled with the generative mechanism itself, ensuring that all synthesized circuits inherently carry consistent and verifiable ownership information.

Under these requirements, existing post hoc, circuit-centric watermarking methods become inadequate for two key reasons. First, they rely on explicit modifications to circuit structures or gate sequences, which may disrupt the generative distribution and potentially degrade the quality or fidelity of generated circuits. Second, their robustness is limited when circuits undergo common post-generation operations such as optimization, editing, or redistribution, making them unreliable for large-scale, automated generative settings. These challenges call for watermarking mechanisms that are natively integrated into quantum circuit generative models, enabling model-level copyright protection while preserving generation quality and robustness.

To address this emerging need, we present the first watermarking framework Q-Tag tailored to QCGMs, establishing a foundation for scalable and model-level copyright protection in AI-driven quantum design. Our method embeds ownership signals directly into the generative process by controlling the latent initialization through a symmetric sampling mechanism, ensuring statistical alignment with the model’s Gaussian prior and avoiding generation artifacts. Furthermore, to ensure robustness against watermark attacks, we introduce a synchronization restoration strategy that corrects latent drift via proactive zero-padding techniques. Empirical evaluations across multiple scenarios demonstrate that our method consistently preserves generation fidelity while enabling reliable watermark detection under various adversarial conditions.

\begin{itemize}[]
\item We are the first to focus the pioneering demand of QCGM copyright protection and proposed an effective watermarking mechanism for potential IP risk in quantum cloud platforms.

\item The proposed mechanism can sufficiently satisfy two properties of QCGM watermarking that are losslessness and robustness. Especially for robustness, we first investigate the satisfy the robustness of four unique synchronization attacks in circuit editing process.

\item Extensive experiments show that the proposed method achieves superior performance in both losslessness and robustness. In particular, the watermark detection rate can reach 99.3\% even under modification attacks.
\end{itemize}

\begin{figure}[t]
\centering
\includegraphics[width=1.0\linewidth]{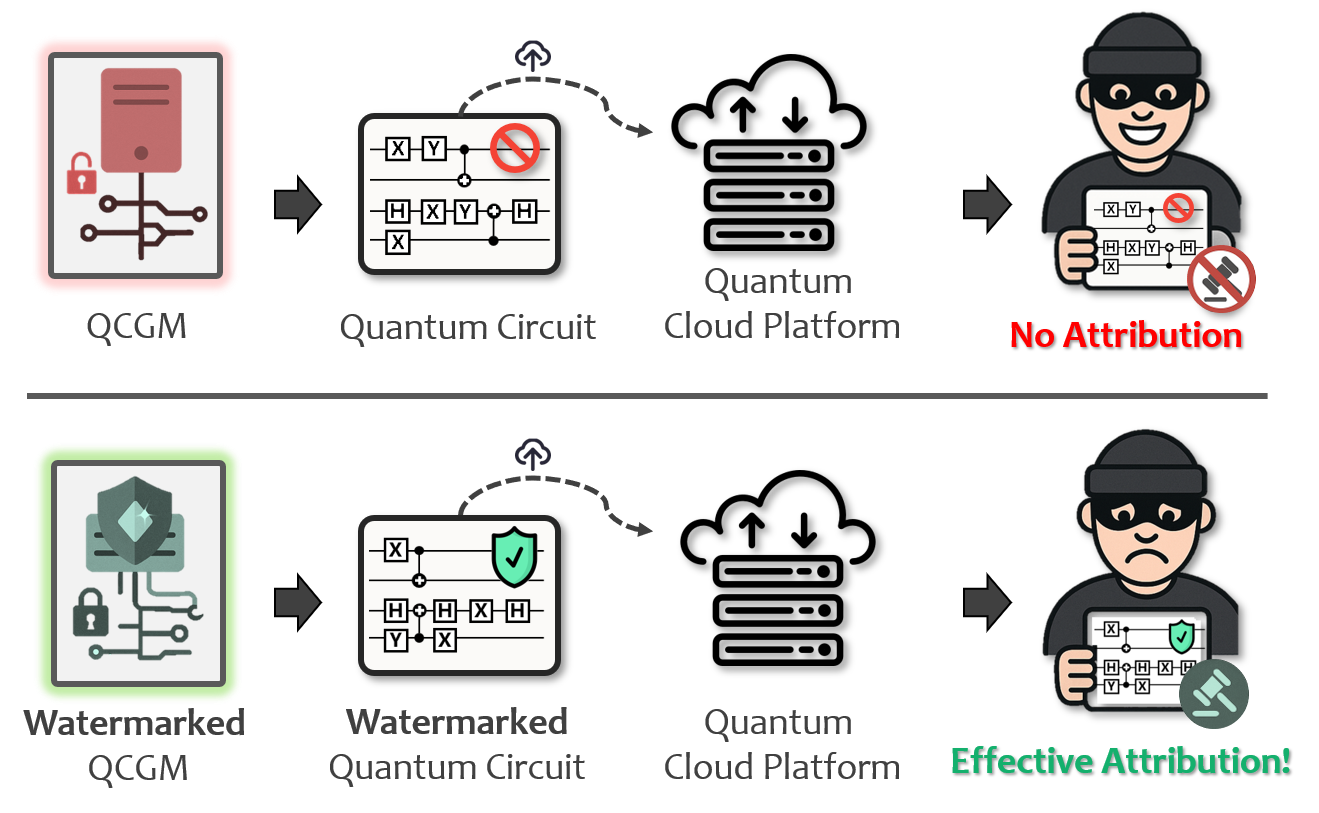}
\caption{\textbf{ Watermarking for quantum circuit generative model.} \textit{Top}: Standard QCGMs generate circuits without embedded ownership, enabling misuse via cloud deployment with no attribution. \textit{Bottom}: Our watermarked QCGM embeds verifiable ownership into each generated circuit, enabling attribution and accountability even after unauthorized redistribution.}
\label{fig_1}
\end{figure}

\section{Related Works}\label{related works}
\subsection{Quantum Circuit Synthesis}
Quantum computing is widely regarded as a transformative technology with immense potential and promising applications in a wide range of fields from basic research in cryptography~\cite{cryptography} and chemistry~\cite{chemistry} to machine learning~\cite{machinelearning} and optimization~\cite{QAQA}.
Despite the continuous advancements in quantum processor technology, the Noisy Intermediate-Scale Quantum (NISQ)~\cite{NISQ} era still confronts numerous challenges, primarily stemming from the inherent limitations of the current devices.
Therefore, efficient quantum circuits must be designed based on available resources and gate sets, which requires further research and expertise to explore optimal gate design and circuit construction strategies.

This task is often referred to as quantum circuit synthesis and mainly includes the following aspects: In order to generate the quantum state of the target output under the conditions of a given quantum processor, a series of basic elements (called quantum gates) need to be designed to generate the quantum state from the reference or standard input state~\cite{genQC}. 
Similarly, the goal may not be to generate only the desired quantum states but to implement quantum algorithms or more general unitary operations. 
In particular, researchers have recently demonstrated important advances in circuit synthesis by utilizing machine learning techniques, from reinforcement learning to generative models, such as in quantum state preparation~\cite{DigitalQuantum}, prediction~\cite{Prepareansatzvqediffusion}, circuit optimization~\cite{quantumcircuitoptimizationdeep}, or unified compilation~\cite{HybridDC}. 

\subsection{Diffusion Model}
Diffusion models are a class of generative models designed to learn the diffusion process that generates a probability distribution for a given training data set. 
The Denoising Diffusion Probability Model (DDPM) introduced by Jonathan Ho et al~\cite{DDPM}. has attracted much attention. It consists of two Markov chains for adding and removing noise.
Its training process is divided into two stages: the initial stage is to gradually add Gaussian noise to the image, which is called the forward process; Subsequently, during the reverse process, the model parameters are trained so that it can learn the inversion of the noise. 
Subsequent studies~\cite{DDIM, Dpm-solve} all adopted this bi-directional chain framework. 
In order to reduce computational complexity and improve efficiency, the latent diffusion model (LDM)~\cite{latentdiffusionmodels} came into being, in which the diffusion process takes place in the latent space. 

With the emergence of potentially diffused LDM and its flexible guidance and regulation in generating capacity~\cite{classifier}, diffusion models have become the preferred method for image generation~\cite{images}, audio synthesis~\cite{diffwave}, video generation~\cite{video}, and protein structure prediction~\cite{Denovo}. 
Recently, a quantum circuit generation model~\cite{genQC} has been proposed to retrain the diffusion model by encoding a large quantum circuit data set. 
After training, a tensor of random noise is input into the model. 
The model iteratively de-noises under the guidance of selected conditions and generates high-quality quantum circuit samples after several steps.

\subsection{Generative Watermarking}
Watermarking effectively solves copyright protection and content authentication by embedding copyright or traceable identification information in carrier data. 
Generative watermarking is a technology used to identify and protect the copyright of related content in generative tasks. 
Its concept comes from the new application of digital watermarking in generative tasks. Different from traditional watermarking and deep learning-based watermarking, on the one hand, generative watermarking can embed the watermark information in the original data set to prevent the data set from being used to train the generative model without authorization;
On the other hand, it is also possible to combine the watermark embedding with the generation process, so that the watermark is integrated into the generated content, thus protecting the copyright of the generated content. 
With the breakthrough of the diffusion model in the field of high-quality content generation, researchers have gradually turned their attention to the generated content protection strategy based on the diffusion model.
To meet this demand, existing watermarking methods include post-processing~\cite{process1, process2, process3, process4, process5, process6}, fine-tuning generation models~\cite{TheStableSignature, fine-tuning3, fine-tuning4, fine-tuning5, fine-tuning6}, and implicit watermark embedding~\cite{latent1, latent2, latent3, GaussianShading}. 
Most of the existing generation model watermarking techniques are designed with images as the carrier. 
Faced with the novel carrier of the quantum circuit, the existing watermarking methods can not apply the correlation generation task well. 
Therefore, this paper deeply explores and effectively solves the problems faced by the model watermarking generation of quantum circuits.

\subsection{Watermarking of Quantum Circuits}
Digital watermarking has been proposed as a method to protect the copyright of quantum circuits while preserving their functionality.
Existing methods for watermarking quantum circuits can be broadly categorized into two main types.
The first type embeds watermarks during the deployment of the circuits, such as the decomposition, mapping, and scheduling stages of quantum circuit.~\cite{Yang2024Multistage}.
These methods rely on auxiliary information introduced during the deployment phase; consequently, they offer no copyright protection if the circuit is compromised or copied prior to deployment.
The second type, in contrast, embeds watermarks in a circuit by adding additional quantum gates, modifying the structure while maintaining its functional integrity~\cite{Vedika2021Decomposition,Roy2025Watermarking}.
These approaches typically embed watermarks by inserting redundant quantum gates, which incurs additional circuit overhead and fails to account for robustness against common post-generation transformations—such as optimization passes, noise-induced distortions, or adversarial tampering.
In summary, existing watermarking schemes are either deployment-dependent or structurally intrusive, and neither paradigm ensures reliable, robust, and intrinsic ownership protection in generative or pre-deployment settings

% \begin{figure*}[!t]
% \centering
% \includegraphics[width=0.9\linewidth]{changes_0519.png}
% % \caption{Changes in the encoded latent.}
% \caption{Effects of circuit modifications on the encoded matrix. Edits to quantum circuits may alter the encoded matrix either by changing specific values or shifting its structural layout. Subfigures illustrate the impact of different edit operations: (a) replacement, (b) appending, (c) insertion, and (d) deletion. Each operation disrupts the latent of the quantum circuits, introducing distortions that may hinder accurate watermark extraction.}
% \label{fig_changes}
% \end{figure*}

\section{Method}\label{Method}
% \noindent\textbf{\large{}Watermark Embedding, Attacks and Extraction}{\large\par}
\begin{figure*}[!t]
\centering
\includegraphics[width=1.0\textwidth]{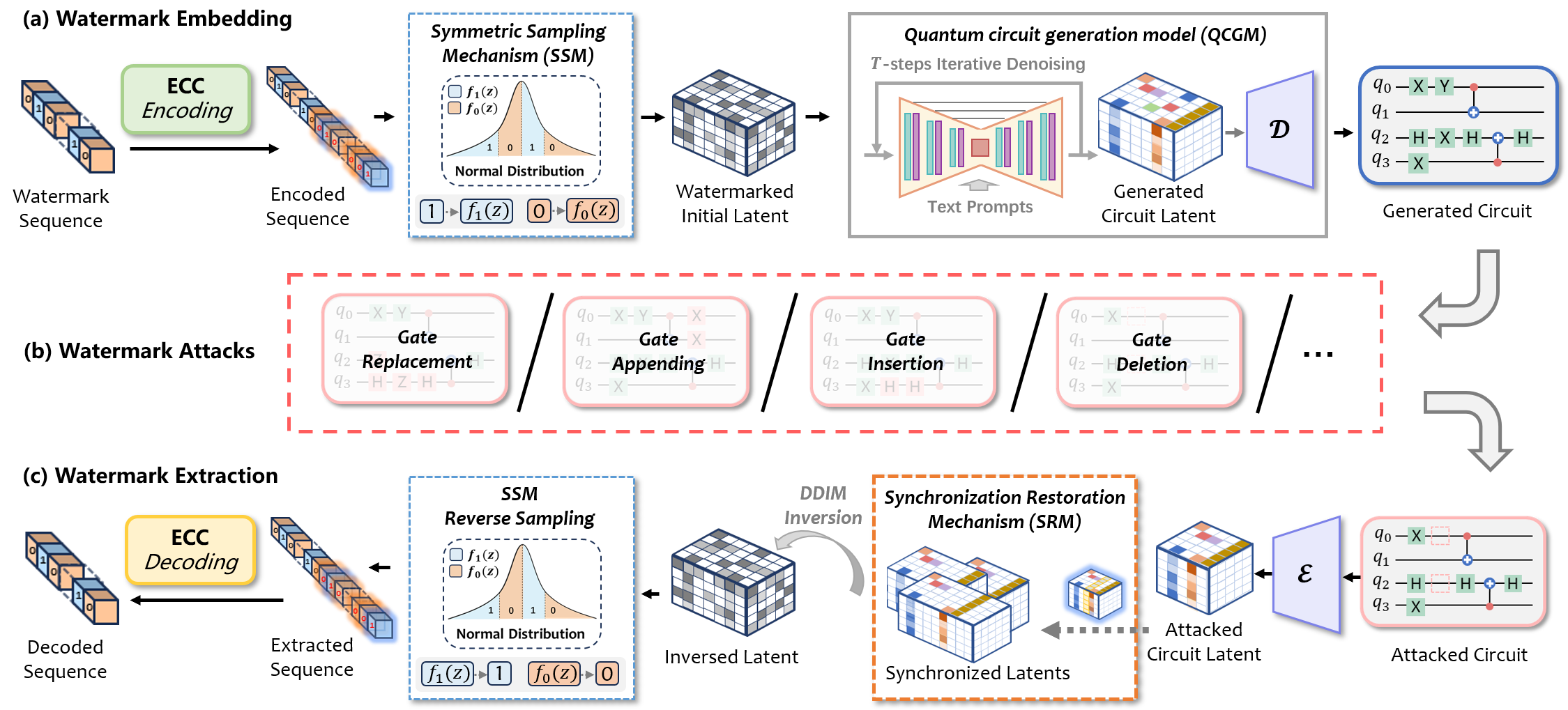}
\caption{Watermark embedding, attacks and extraction process for quantum circuits generative model (QCGM). The watermark is embedded into the starting latent of QCGM through SSM, then the watermarked QC is generated and potentially attacked. The SRM and SSM reverse sampling is applied in the extraction stage to extract the watermark.}
\label{fig_2}
\end{figure*}
A complete overview of our proposed pipeline is provided in Fig.~\ref{fig_2}, comprising three main stages: (a) \textit{watermark embedding phase}, where a binary watermark is injected into the latent input of the QCGM via a symmetric sampling mechanism (SSM); (b) \textit{watermark attacks phase}, where the generated quantum circuit may be subject to structural manipulations intended to remove or disrupt the watermark; and (c) \textit{watermark extraction phase}, where the watermark is recovered and verified from the possibly tampered circuit. During extraction, a synchronization restoration mechanism (SRM) is additionally applied to correct latent misalignments introduced by circuit-level edits prior to watermark decoding.

\subsection{Watermark Embedding.}
To embed a watermark, typically represented as a binary sequence, we first apply error-correcting codes (ECC) to introduce redundancy and enhance robustness against bit-level distortions. The encoded sequence is then spectrally transformed via the symmetric sampling mechanism, producing an initial latent vector that satisfies two essential criteria: (i) compatibility with the QCGM input format, which statistical conformity to the standard Gaussian prior, and (ii) fidelity to the embedded watermark signal. This watermarked latent is subsequently processed by the QCGM through iterative denoising steps, ultimately yielding a quantum circuit with imperceptibly embedded ownership information. The success of this embedding procedure hinges on synthesizing latent vectors that are both semantically meaningful and statistically indistinguishable from standard QCGM inputs, while faithfully encoding the intended watermark. 
%Detailed implementation procedures are described in the Methods section.
The detailed implementation procedures are as follows:

%\subsection{Watermark Embedding Procedure}

To embed a $k$-bit binary sequence $m \in \{0,1\}^k$, i.e., the watermark, the sequence is first encoded using an error correction coding mechanism \( \mathbf{Enc} \) into an $n$-bit codeword \( \mathcal{S}_{\text{en}} = \mathbf{Enc}(m) \), where $n = f_c \times f_h \times f_w$ and $f_c$, $f_h$, and $f_w$ denote the dimensions of the input latent. The encoded sequence \( \mathcal{S}_{\text{en}}=\{s_{\text{en}}^1,s_{\text{en}}^2,\dots,s_{\text{en}}^n\} \) satisfies two key properties: \textbf{pseudorandomness}, meaning that \( \mathcal{S}_{\text{en}} \) is computationally indistinguishable from a uniformly distributed random sequence; and \textbf{error-correctability}, ensuring that the original message $s$ can be reliably recovered even when \( \mathcal{S}_{\text{en}} \) experiences a bounded number of bit errors. Next, SSM is applied to \( \mathcal{S}_{\text{en}} \) to generate the watermarked starting latent \( \mathcal{Z}_T \). This process yields the starting latent \( \mathcal{Z}_T = \{ z_T^1, z_T^2, \dots, z_T^n \} \). Subsequently, \( \mathcal{Z}_T \) is passed through QCGM for $T$ diffusion steps to produce the final latent \( \mathcal{Z}_0 \). Finally, the watermarked quantum circuit \( Q_m \) is generated via the trained decoder \( \mathcal{D} \), following \( Q_m = \mathcal{D}(\mathcal{Z}_0) \).

In this work, the encoding function \( \mathbf{Enc} \) is instantiated using a lightweight combination of \textit{repetition coding} and \textit{stream cipher encryption}, designed to balance robustness and pseudorandomness. Given a \( k \)-bit binary message \( m \in \{0,1\}^k \), we first apply repetition coding by duplicating each bit \( v \) times to obtain a length-\( n \) intermediate sequence \( M \in \{0,1\}^n \), where \( n = k \times v \). This repetition enhances error tolerance by introducing redundancy at the bit level.

Next, we generate a binary pseudorandom key \( K \in \{0,1\}^n \) using a stream cipher, such as a seeded pseudorandom number generator. The final encoded sequence \( \mathcal{S}_{\text{en}} \) is then computed by applying bitwise XOR between the repeated message and the key:
\[
\mathcal{S}_{\text{en}} = M \oplus K
\]

This construction ensures that \( \mathcal{S}_{\text{en}} \) exhibits \textit{pseudorandomness}, rendering it statistically indistinguishable from a uniformly random sequence, while the underlying redundancy introduced by repetition supports \textit{error-correctability} during watermark recovery. The encoded sequence \( \mathcal{S}_{\text{en}} \) is subsequently used to guide the symmetric sampling process for watermark embedding. The whole embedding algorithm is shown in Algorithm 1.

\begin{algorithm}[t]
\caption{Watermark Embedding}
\label{alg:ssm_embedding}
\KwIn{\small Watermark sequence $m \in \{0,1\}^k$; QCGM model $\mathcal{G}$; Decoder $\mathcal{D}$; ECC Encoder $\mathbf{Enc}$; Number of diffusion steps $T$.}
\KwOut{\small Watermarked quantum circuit $Q_m$.}

\textbf{1. Encode watermark sequence:}\\
$\mathcal{S}_{\text{en}} = \mathbf{Enc}(m)$ \tcp*{Apply error correction coding, obtain $n$-bit codeword}

\textbf{2. Symmetric sampling to generate starting latent:}
\For{$i = 1$ \textbf{to} $n$}{
    Sample $z^i \sim \mathcal{N}(0, I)$ \tcp*{Sample from standard Gaussian}
    Partition $\mathcal{N}(0,I)$ into $P_0$ and $P
    _1$ \tcp*{Symmetric partition}
    \eIf{$s_{\text{en}}^i = 0$}{
        \textbf{Adjust} $z^i$ to satisfy $P_0$
    }{
        \textbf{Adjust} $z^i$ to satisfy $P
    _1$
    }
    Set $z_T^i = z^i$
}
Construct starting latent $\mathcal{Z}_T = \{z_T^1, z_T^2, \dots, z_T^n\}$

\textbf{3. Diffusion process:}\\
Apply $T$-step diffusion to $\mathcal{Z}_T$ using QCGM $\mathcal{G}$ to obtain diffused latent $\mathcal{Z}_0$

\textbf{4. Decode quantum circuit:}\\
Generate quantum circuit $Q_m = \mathcal{D}(\mathcal{Z}_0)$

\Return $Q_m$
\end{algorithm}

\subsection{Watermark Attacks.}
In realistic deployment settings, it is crucial to account for an adversarial attacker who, upon identifying the existence of a watermark, may attempt to erase it while preserving circuit functionality. Here, we consider four representative attack operations: gate replacement, appending, insertion, and deletion; each of which subtly alters the structure of the quantum circuit without compromising its computational correctness (Fig.~\ref{fig_changes}).

% \vspace{0.3cm}
% \noindent \textit{Replacement.} 
\subsubsection{Replacement}
Functionally equivalent gate substitu-tions, such as replacing a Pauli-X gate with a Hadamard
–Z–Hadamard sequence (Fig.~\ref{fig_changes}a), preserve circuit behavior but alter its internal structure, introducing mismatches with the expected watermark spectral signature.

% \vspace{0.3cm}
\subsubsection{Appending} 
Appending gates to the end of unmeasured qubit lines (Fig.~\ref{fig_changes}b) can leave final measurement outcomes unaffected, yet distort spatial encodings critical for watermark recovery.

% \vspace{0.3cm}
\subsubsection{Insertion} 
Mid-circuit insertions of neutral gate pairs, such as two Hadamard gates satisfying $U^\dagger U = I$, and in particular $U^\dagger = U$ (Fig.~\ref{fig_changes}c), introduce logically redundant operations that can confound alignment with the embedded watermark structure.

% \vspace{0.3cm}
\subsubsection{Deletion} 
Removing non-critical gates (e.g., a Pauli-Z on $q_0$) can similarly preserve circuit output (Fig.~\ref{fig_changes}d), while degrading the positional coherence necessary for robust watermark decoding.

\begin{algorithm}[t]
\caption{Watermark Extraction}
\label{alg:extraction_SRM}
\KwIn{\small Distorted quantum circuit $Q_d$; Encoder $\mathcal{E}$; ECC Decoder $\mathbf{Dec}$; Threshold $\tau$; Number of diffusion steps $T$.}
\KwOut{\small Extracted watermark $\hat{m}$ or detection result.}

\textbf{1. Encode the circuit:}\\
Obtain latent representation $\hat{\mathcal{Z}}_0 = \mathcal{E}(Q_d)$

\textbf{2. Apply DDIM inversion:}\\
Perform $T$-step DDIM inversion on $\hat{\mathcal{Z}}_0$ to get inverted latent $\hat{\mathcal{Z}}_T$

\textbf{3. Standard Extraction:}\\
Perform reverse sampling on each element $\hat{z}_T^i$ to generate decoded sequence $\hat{\mathcal{S}}_{de}$\;
Decode watermark $\hat{m} = \mathbf{Dec}(\hat{\mathcal{S}}_{de})$\;

\If{$\mathcal{L}(m, \hat{m}) > \tau$}{
    \Return $\hat{m}$ (Watermark detected)
}

\textbf{4. Synchronization and Verification (SRM):}
\begin{itemize}
    \item Insert zero-matrix $\mathcal{O}$ of size $f_c \times f_h \times k$ at each column position $i \in [1, f_w]$ into $\hat{\mathcal{Z}}_0$ to generate synchronization sets $\bar{\mathbb{Z}}_0$.
    \item \ForEach{ $\bar{\mathcal{Z}}_0^i\in\bar{\mathbb{Z}}_0$ } {
        Perform $T$-step DDIM inversion on $\bar{\mathcal{Z}}_0^i$ to obtain $\bar{\mathcal{Z}}_T^i$\;
        
        Perform reverse sampling to generate $\hat{\mathcal{S}}_{de}^i$\;
        
        Decode $\hat{m}^i = \mathbf{Dec}(\hat{\mathcal{S}}_{de}^i)$\;
        
        \If{$\mathcal{L}(m, \hat{m}^i) > \tau$}{
            \Return $\hat{m}^i$ (Watermark detected)
        }
    }
\end{itemize}

\Return No valid watermark detected.
\end{algorithm}

\subsection{Watermark Extraction.}
To recover the embedded watermark, we first encode the attacked circuit back into the latent domain. A synchronization restoration mechanism (SRM) is then applied to the perturbed latent, generating a set of candidate latents that are better aligned with the original sampling structure. This step mitigates distortions introduced by structural edits and enhances alignment between the extracted latent and the expected watermark pattern. Subsequently, we apply denoising diffusion implicit model (DDIM) inversion to estimate the original watermarked latent from the synchronized candidates. Finally, a reverse sampling procedure, mirroring the SSM used during embedding, is performed, followed by error-correcting code decoding to reconstruct the original watermark bit sequence.
The detailed implementation procedures are as follows:

%\subsection{Watermark Extracting Procedure}

To extract the watermark from the generated quantum circuit $Q_m$, we first encode $Q_m$ using the quantum circuit encoder $\mathcal{E}$ to obtain the latent representation: $\hat{\mathcal{Z}}_{0} = \mathcal{E} (Q_m)$. Subsequently, diffusion model inversion~\cite{DDIM} is applied to $\hat{\mathcal{Z}}_0$ over $T$ steps to recover the inverted latent representation $\hat{\mathcal{Z}}_T$, which approximates the original watermarked latent $\mathcal{Z}_T$. For each element $\hat{z}_T^i \in \hat{\mathcal{Z}}_T$, a reverse sampling procedure is performed to obtain a binary sequence $\hat{\mathcal{S}}_{\text{de}}$. Specifically, if $\hat{z}_T^i \in P_0$, the $i$-th bit $\hat{s}_{\text{de}}^i$ is set to 0; otherwise, it is set to 1. The final extracted watermark $\hat{m}$ is obtained by decoding the sequence $\hat{\mathcal{S}}_{\text{de}} = \{ \hat{s}_{\text{de}}^1, \hat{s}_{\text{de}}^2, \dots, \hat{s}_{\text{de}}^n \}$ using the decoding mechanism $\mathbf{Dec}$ corresponding to the encoder $\mathbf{Enc}$, such that $\hat{m} = \mathbf{Dec}(\hat{\mathcal{S}}_{\text{de}})$.
The extracted watermark $\hat{m}$ is then compared with the original watermark $m$ for verification. If the similarity $\mathcal{L}(m, \hat{m})$ exceeds a predefined threshold $\tau$, the generated quantum circuit $Q_m$ is identified as successfully watermarked with $m$. In this work, the \( \mathbf{Dec} \) procedure can be described as a two-step process. First, the encrypted sequence \( \hat{\mathcal{S}}_{\text{en}} \) is decrypted via bitwise XOR with the known pseudorandom key $K$. Then, majority voting is applied over every \( v \)-bit repetition block to reconstruct the original \( k \)-bit watermark \( \hat{m} \).

We call the above extracting procedure as standard extraction. Then to address the latent misalignment caused by structural perturbations, we propose a synchronization restoration mechanism, which aims to proactively correct positional desynchronization and enable reliable watermark extraction under adversarial conditions. 

\vspace{0.2cm}
\noindent \textit{Synchronization Restoration Mechanism.} Given a distorted quantum circuit $Q_d$, we first obtain its latent representation $\hat{\mathcal{Z}}_0 = \mathcal{E}(Q_d)$ with shape $f_c \times f_h \times f_w$. To compensate for potential misalignment, we construct a zero-initialized matrix $\mathcal{O} \in \mathbb{R}^{f_c \times f_h \times k}$ and systematically inject it at different column positions. For each candidate index $i \in [1, f_w]$, we insert $\mathcal{O}$ at the $i$-th column to generate a compensated latent representation $\bar{\mathcal{Z}}_0^i$, defined as 
\begin{equation} \bar{\mathcal{Z}}_0^i = [\hat{\mathcal{Z}}_0[1:i]; \mathcal{O}; \hat{\mathcal{Z}}_0[i:f_w]]. 
\end{equation} 
This process yields a set of compensated latents $\bar{\mathbb{Z}}_0 = {\bar{\mathcal{Z}}_0^i}, {i=1,2,\dots,f_w}$. We then apply the standard extraction procedure to each $\bar{\mathcal{Z}}_{0}^i \in \bar{\mathbb{Z}}_0$. If any candidate latent leads to a recovered watermark $\hat{m}$ that satisfies the similarity criterion $\mathcal{L}(m, \hat{m}) > \tau$, the circuit $Q_d$ is deemed successfully watermarked.

This strategy hinges on the assumption that $\mathcal{L}(m, \hat{m})$ exceeds the threshold $\tau$ with high probability when alignment is correctly restored. Conversely, in cases of failed compensation or absence of a watermark, the likelihood of false detection remains negligibly low, thus ensuring robust and reliable verification.

\subsection{Robustness Evaluation Metric}  

We evaluate robustness using TPR as the primary metric. The calculation of TPR can be described as follows: For the circuit with watermark $s$, we extract the potential watermark sequence $\hat{s}$ from the circuit. Then we calculate the similarity of $s$ and $\hat{s}$, denoted as $\mathcal{L}(s, \hat{s})$. In our experiment, $\mathcal{L}$ is bitwise similarity, which is defined as the fraction of positions at which the corresponding bits are identical. Formally, it is given by:
$$
\mathcal{L}(s, \hat{s}) = \frac{1}{n} \sum_{i=1}^{n} (s_i = \hat{s}_i),
$$
where the Boolean expression $(s_i = \hat{s}_i)$ evaluates to 1 if the two bits are equal and 0 otherwise, and $n$ is the length of $s$ and $\hat{s}$. A similarity score of 1 indicates perfect agreement, while a score approaching 0 reflects increasing dissimilarity. A watermarked circuit is considered successfully identified if its similarity score exceeds a predefined threshold $\tau$, which we set to 0.7916. This threshold is derived via hypothesis testing to ensure a fixed false positive rate (FPR) of $10^{-3}$.
%detailed methodology is provided in the Supplementary Materials. 
Since each circuit yields a single similarity score, it contributes exactly one instance to the TPR computation. To ensure statistical reliability and broader evaluation coverage, we conduct experiments over 1000 independently generated watermarked circuits. The final true positive rate is computed as: 
$$
\text{TPR} = \frac{N_w}{1000}
$$
where $N_w$ denotes the number of circuits for which $\mathcal{L}(s, \hat{s})\geq \tau$.

\begin{figure*}[!t]
\centering
\includegraphics[width=0.9\linewidth]{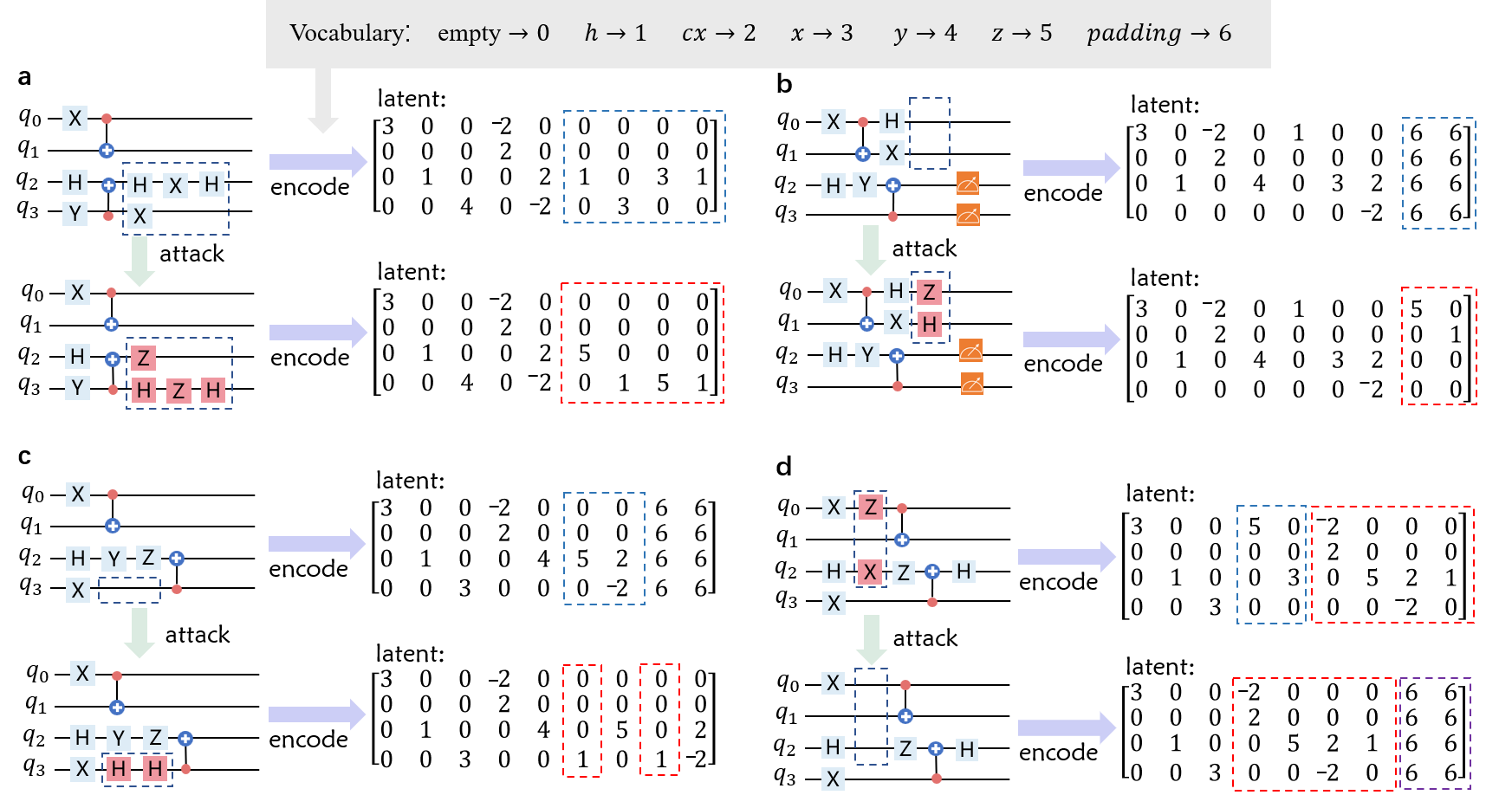}
% \caption{Changes in the encoded latent.}
\caption{Effects of circuit modifications on the encoded matrix. Edits to quantum circuits may alter the encoded matrix either by changing specific values or shifting its structural layout. Subfigures illustrate the impact of different edit operations: (a) replacement, (b) appending, (c) insertion, and (d) deletion. Each operation disrupts the latent of the quantum circuits, introducing distortions that may hinder accurate watermark extraction.}
\label{fig_changes}
\end{figure*}

\section{Experiments}\label{Experiments}
We assess the proposed method along two key dimensions: fidelity and robustness. Fidelity is evaluated through both theoretical analysis and empirical results, demonstrating that the watermark embedding process preserves the generation quality of quantum circuits. Robustness is measured using the true positive rate (TPR), the proportion of correctly identified watermarked circuits among all watermarked instances, which quantifies the reliability of watermark detection under various conditions. To further evaluate generalizability, we analyze how TPR varies with the number of inference and inversion steps. An ablation study is also conducted to isolate the impact of the SRM, highlighting its contribution to enhancing resilience against adversarial watermark attacks.

\subsection{Implementation Details}
In this work, we adopt the quantum circuit generative model~\cite{genQC} as the foundation for our experiments. The model operates within a latent space of dimensionality $4 \times 8 \times 48$, corresponding to quantum circuits composed of 8 qubits with up to 48 quantum gates. During generation, we set the guidance scale to 7.5 and the denoising steps to 50. An empty prompt is used to initiate the generation process, ensuring unbiased latent initialization. For the inversion procedure, we similarly adopt 50 inversion steps as the default setting. The message to be embedded is a 24-bit binary sequence. All experiments are implemented using the PyTorch framework and executed on a single NVIDIA RTX 4090 GPU.

\subsection{Parameter Studies}
\subsubsection{Threshold \texorpdfstring{$\tau$}{\text{tau}}}
In this section, we will introduce how the threshold $\tau$ is determined according to the false positive rate. Recap the extraction process: 
For a specific watermark $m$, the proposed Q-Tag model starts with the latent sampled with $\mathcal{S}_{en}$ to embed the watermark identity information into every generated quantum circuit. Then, assuming $\hat{m}$ is the extracted watermark, the similarity of the two watermark denoted $\mathcal{L}(m, \hat{m})$, is compared with a threshold value $\tau$, when $\mathcal{L}(m, \hat{m})\geq \tau$, the quantum circuit is identified as watermarked.

As discussed in \cite{RootingDeepfake}, it is usually assumed that the extracted watermark bit $\hat{s}_{1}, \hat{s}_{2},..., \hat{s}_{k}$ is independent and equally distributed, where $\hat{s}_{i}$ follows a Bernoulli distribution with parameter 0.5. In this case, $\mathcal{L}(m, \hat{m})$ follows a binomial distribution with parameter $(k,0.5)$. 

Once the distribution of $\mathcal{L}(m, \hat{m})$ is established, the FPR is defined as the probability that the unwatermarked quantum circuit, $\mathcal{L}(m, \hat{m})$, surpasses the threshold $\tau$.
This probability can be computed using the regularized incomplete beta function 
$B_x(a;b)$.
Further elaboration of this expression can be found in \cite{TheStableSignature}.

\begin{equation}
\label{eq_7}
\begin{split}
FPR(\tau) &= \text{Pr}(\text{Acc}(s, \hat{s}) > \tau) = \frac{1}{2^k} \sum_{i=\tau+1}^{k} \binom{k}{i}\\
&= B_{\frac{1}{2}}(\tau+1, k-\tau)
\end{split}
\end{equation}

For the threshold, we embed a random watermark $m$ into QCGM, generate 10000 quantum circuits, and use the test of $\mathcal{L}(m, \hat{m})$.
We report the tradeoff between the TPR and the FPR, while varying $\tau\in\{0,..,24\}$. The TPR is measured directly.
The FPR is inferred from Eq.\ref{eq_7}.
The experiment is run on 10 random watermarks and we report averaged results.
The experimental results are presented in Fig.\ref{TPR_FPR}. A lower FPR inevitably leads to a reduction in TPR, as stricter detection criteria allow for fewer bit errors and result in a higher threshold for successful extraction. Nevertheless, Q-Tag consistently maintains a TPR above 90\% even under the stringent setting of FPR = $10^{-3}$, demonstrating reliable watermark detectability under conservative verification constraints.

\begin{figure*}[!t]
\centering
\includegraphics[width=6.0in]{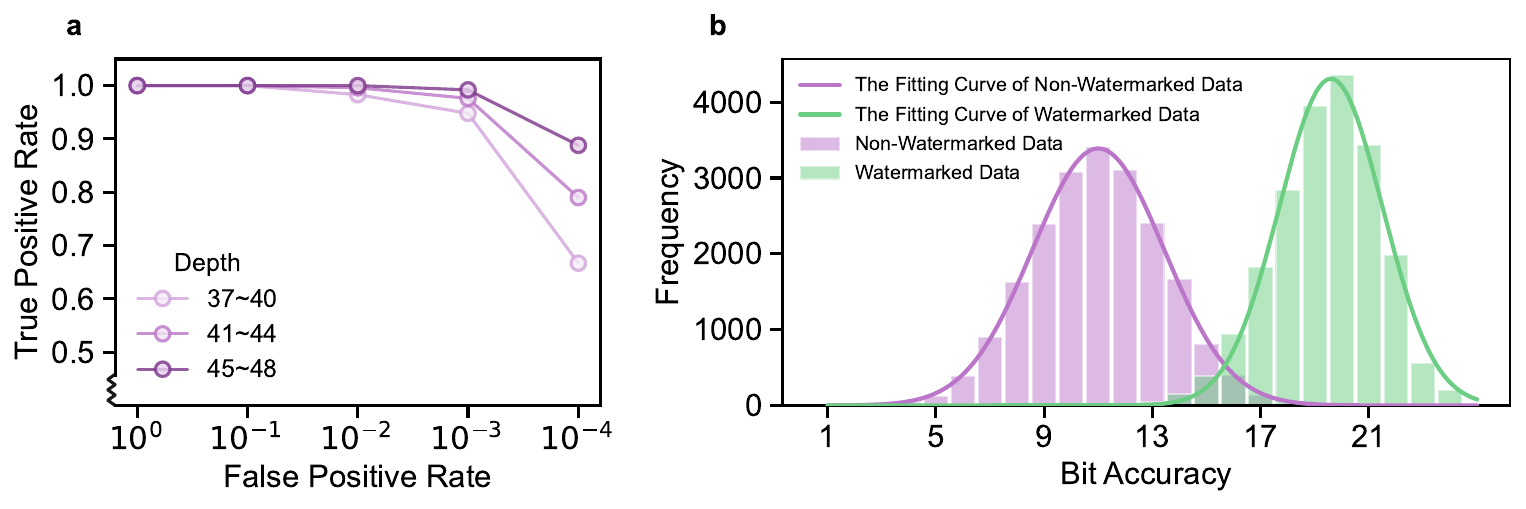}
\caption{Experiment on threshold selection. (a) TPR/FPR. (b) Bit accuracy distribution.}
\label{TPR_FPR}
\end{figure*}

\subsubsection{False Positive of Synchronization Restoration Mechanism}
In our extraction process, a SRM is employed to recover potential misalignments in the latent space. While this enhances robustness, it also increases the extraction payload, which may inadvertently raise FPR. To mitigate this, we introduce a similarity threshold $th$ on the similarity score $\mathcal{L}(m,\hat{m})$, ensuring that only highly confident matches are accepted. This section describes how we determine an appropriate value for $th$.

To calibrate the threshold, we conduct a statistical analysis based on Neyman–Pearson theory. Specifically, we fix a watermark message 
$m$, and generate 20,000 watermarked quantum circuits under various settings. In parallel, another 20,000 circuits are generated using random latent initializations, without any embedded watermark. We apply the full extraction pipeline to both sets and record the similarity scores between $m$ and the extracted message $\hat{m}$, producing two empirical distributions.

As shown in Fig.~\ref{TPR_FPR}(b), the resulting distributions of bit-level accuracy for the watermarked and non-watermarked cases approximately follow Gaussian profiles, but with distinct means and variances. We then apply hypothesis testing to determine the decision threshold $th$, defined under the following hypotheses:

\begin{itemize}
\item 
$H_0$: The circuit does not contain a watermark (null hypothesis).
\item 
$H_1$: The circuit contains a valid watermark (alternative hypothesis).
\end{itemize}

The likelihood function of these two hypotheses can be written as:

\begin{equation}
\label{eq_H0}
f\left( x \vert H_0 \right) = \frac{1}{\sigma_0\sqrt{2\pi}} \cdot e^{-\frac{(x - \mu_0)^2}{2\sigma_0^2}}
\end{equation}

\begin{equation}
\label{eq_H1}
f\left( x \vert H_1 \right) = \frac{1}{\sigma_1\sqrt{2\pi}} \cdot e^{-\frac{(x - \mu_1)^2}{2\sigma_1^2}}
\end{equation}
In the case of Neyman-Pearson (NP) approach in signal detection, the decision rule is defined as

\begin{equation}
\label{eq_HH}
\frac{f(x | H_1)}{f(x | H_0)} \overset{H_0}{\underset{H_1}{\gtrless}} \xi
\end{equation}
that is

\begin{equation}
\label{eq_HH01}
x \overset{H_0}{\underset{H_1}{\gtrless}} g(\xi)
\end{equation}
where $\xi$ is the decision threshold and $g(\xi)$ is a function of $\xi$  which is determined by Eqs.(\ref{eq_H0}),(\ref{eq_H1}),(\ref{eq_HH01}). Let $th = g(\xi)$, the decision rule is equivalent to

\begin{equation}
\label{eq_th}
x \overset{H_0}{\underset{H_1}{\gtrless}} th
\end{equation}
Then, according to the given probability of false-alarm (denoted by $\alpha_0$), which is defined as

\begin{equation}
\label{eq_a0}
\alpha_0 = \int_{\text{th}}^{\infty} f(x | H_0) \, dx
\end{equation}
Based on the empirical distributions, we compute the detection threshold $th$ corresponding to a target false positive rate $\alpha_0=10^{-3}$. Using the fitted parameters-$\mu_0$= 9.00, $\sigma_0$ = 2.43, $\mu_1$= 18.60, $\sigma_1$ = 1.91, we apply Neyman–Pearson analysis to determine the optimal threshold. According to Eq.~\ref{eq_a0}, the resulting value is $th=16.5$, meaning that circuits with fewer than 16.5 accurately extracted bits are classified as non-watermarked.

In our implementation, we conservatively set the final verification threshold $\tau=0.7916$, which corresponds to 19 correctly extracted bits—well above the decision boundary of $th=16.5$. This ensures that both the compensation process and watermark detection maintain a false positive rate below $10^{-3}$, thus preserving detection reliability under practical deployment conditions.

\subsection{Fidelity evaluation}\label{Fidelity evaluation}

Achieving high-fidelity watermark embedding is nontrivial, particularly in the context of QCGMs that rely on strict statistical assumptions about their inputs. In QCGMs, the generation process begins from a starting latent vector sampled from a standard Gaussian distribution. This latent serves as the origin of the diffusion trajectory and directly determines the structure of the resulting quantum circuit. Embedding watermark information into this latent is appealing, as it provides a direct and global influence over the output. However, naively modifying the latent to encode information risks violating the Gaussian prior, which may lead to distributional mismatch, degraded generation quality, or even failure modes in the denoising process.

Our motivation, therefore, is to enable watermark embedding at this sensitive entry point while preserving the statistical integrity of the latent space. To this end, we propose a symmetric sampling mechanism that injects watermark information through controlled perturbations that maintain exact adherence to the standard Gaussian distribution. This distribution-preserving design ensures compatibility with the QCGM’s generative assumptions, allowing the watermark to be imperceptibly embedded without compromising fidelity. By addressing the tension between expressiveness and distributional correctness, our approach enables robust and high-quality watermarking within diffusion-based generative pipelines.

In this section, we evaluate the effectiveness of the proposed SSM with respect to fidelity. We begin by outlining the design of SSM. Given a standard Gaussian distribution \( f(z) \), we split it into two symmetric partition \( P_0 \) and \( P_1 \), which are mirror-symmetric about $z=0$. For each bit in the binary watermark sequence to be embedded \( s_{\text{en}}^i \), we sample a value \( z^i \sim \mathcal{N}(0, I) \). If \( z^i \in P{s_{\text{en}}^i} \), we set \( z_T^i = z^i \); otherwise, we flip its sign and set \( z_T^i = -z^i \). For instance, if \( s_{\text{en}}^i = 0 \), then \( P_{s_{\text{en}}^i} = P_0 \), and similarly for \( s_{\text{en}}^i = 1 \).

% \vspace{0.4cm}
% \noindent \textbf{\textit{Theoretical Proof.} }
\subsubsection{Theoretical Proof}
We provide a theoretical proof that each element \( z_T^i \) generated via the SSM follows a standard Gaussian distribution, i.e.,  \( z_T^i \sim \mathcal{N}(0, I)\).
\begin{theorem}[Distribution Preservation of SSM]
Let \( z^i \sim \mathcal{N}(0, I) \) be a standard Gaussian variable, and let \( s_{\mathrm{en}}^i \in \{0, 1\} \) be a pseudorandom bit sampled uniformly at random. Let \( z_T^i \) be the sample generated by the symmetric sampling mechanism (SSM), defined as:
\[
z_T^i =
\begin{cases}
z^i, & \text{if } z^i \in P_{s_{\mathrm{en}}^i}, \\\\
- z^i, & \text{otherwise},
\end{cases}
\]
where \( P_0 \) and \( P_1 \) are symmetric, equiprobable partitions of the Gaussian distribution with respect to the origin (i.e., \( \Pr(z^i \in P_0) = \Pr(z^i \in P_1) = 0.5 \), and \( P_1 = -P_0 \)). Then \( z_T^i \sim \mathcal{N}(0, I) \).
\end{theorem}

\begin{proof}
To show that \( z_T^i \sim \mathcal{N}(0, I) \), it suffices to prove that for any \( x \in \mathbb{R} \), the cumulative distribution function satisfies
\[
\Pr(z_T^i \leq x) = \Pr(y^i \leq x),
\]
where \( y^i \sim \mathcal{N}(0, I) \).

According to SSM, we have
\begin{align}
\Pr(z_T^i \leq x) &= \Pr(z^i \leq x,\, z^i \in P_{s_{\mathrm{en}}^i}) \nonumber \\
&+ \Pr(-z^i \leq x,\, z^i \in P_{1 - s_{\mathrm{en}}^i}) \nonumber \\
&= \Pr(z^i \leq x) \cdot \Pr(z^i \in P_{s_{\mathrm{en}}^i}) \nonumber \\
&+ \Pr(z^i \geq -x) \cdot \Pr(z^i \in P_{1 - s_{\mathrm{en}}^i}) \label{eq:ssm-cdf}
\end{align}

By the symmetry of the Gaussian distribution and the equiprobability of partitions,
\begin{align}
\Pr(z^i \in P_{s_{\mathrm{en}}^i}) &= \Pr(z^i \in P_{1 - s_{\mathrm{en}}^i}) = 0.5, \nonumber \\
\Pr(z^i \geq -x) &= \Pr(z^i \leq x). \nonumber
\end{align}

Substituting into Eq.~\eqref{eq:ssm-cdf} gives
\begin{align}
\Pr(z_T^i \leq x) &= 0.5 \cdot \Pr(z^i \leq x) + 0.5 \cdot \Pr(z^i \leq x) \nonumber \\
&= \Pr(z^i \leq x). \nonumber
\end{align}
which matches the CDF of \( \mathcal{N}(0, I) \). 
Therefore, \( z_T^i \sim \mathcal{N}(0, I) \).
\renewcommand{\qedsymbol}{}
\end{proof}

% \vspace{0.2cm}
% \noindent \textbf{\textit{Empirical Evidence.}}
\subsubsection{Empirical Evidence}
In addition to the theoretical guarantee, we empirically assess fidelity by evaluating the indistinguishability between watermarked and non-watermarked latents. Specifically, we employ t-distributed stochastic neighbor embedding (t-SNE) ~\cite{van2008visualizing} to visualize their distributions in a reduced dimensional space, allowing us to qualitatively examine whether watermark embedding preserves the underlying structure of the latent space. We generate 500 initial latent vectors embedded with watermarks $\mathcal{Z}_T^{\mathcal{W}}$ and 500 without $\mathcal{Z}_T^{\emptyset}$, and and propagate both sets through the generative process to obtain the final latents prior to variational autoencoder (VAE) decoding ($\mathcal{Z}_0^{\mathcal{W}}$ and $\mathcal{Z}_0^{\emptyset}$). These latents are then projected into two-dimensional space using t-SNE, as shown in Fig.~\ref{fig_tsne}(a,b). The resulting visualizations reveal no clear separation between the two distributions, either before or after diffusion. This empirical evidence supports the conclusion that the proposed watermarking strategy preserves the statistical properties of the latent space, introducing no observable distributional bias.

\begin{figure}[!t]
\centering
\includegraphics[width=1\linewidth]{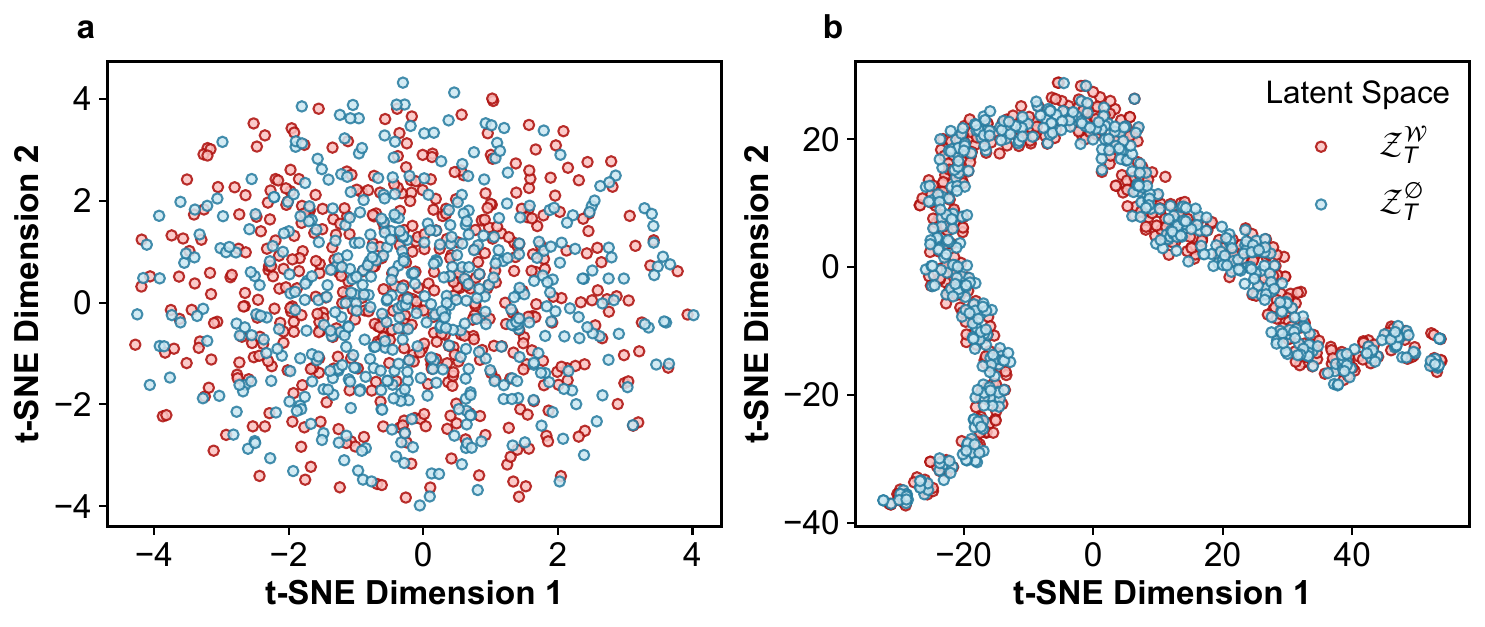}
\caption{Fidelity experiments of Q-Tag generation, (a) the initial starting latent with ($\mathcal{Z}_T^{\mathcal{W}}$) and without ($\mathcal{Z}_T^{\emptyset}$) watermark are indistinguishable through t-SNE; (b) the diffused latent with ($\mathcal{Z}_0^{\mathcal{W}}$) and without ($\mathcal{Z}_0^{\emptyset}$) watermark are also indistinguishable through t-SNE.}
\label{fig_tsne}
\end{figure}

\subsection{Robustness evaluation}\label{Robustness evaluation}

Ensuring robustness against structural attacks poses a significant challenge in watermark recovery. In practice, watermarked quantum circuits may undergo various attacks, such as gate insertion, deletion, that do not alter their functionality but can induce misalignments in the corresponding latent representations. These latent shifts disrupt the spatial correspondence required for reliable watermark extraction, particularly when the watermark is embedded at the input stage of the generative process.

To address this, we introduce a synchronization restoration mechanism designed to mitigate the effects of such structural attacks. The core idea is to prepend a zero-initialized alignment prior to the inversion process, effectively re-anchoring the latent structure and restoring compatibility with the original watermark embedding pattern. This lightweight correction allows the inversion path to converge more reliably to the intended watermarked latent, even under adversarial or noisy conditions. Detailed can be found in Methods section.

We evaluate the robustness of the proposed watermarking framework by measuring the TPR of watermark detection under both clean and adversarial conditions. Specifically, we first assess TPR across different circuit lengths without any attacks, followed by experiments simulating four categories of adversarial watermark erasure.

% \vspace{0.2cm}
\subsubsection{TPR under clean conditions}
Detection performance without distortion is shown in Fig.~\ref{robustness}(a), with the number of quantum gates ranging from 33 to 48. Across all configurations, the proposed method achieves a consistently high TPR ($\geq 90\%$), demonstrating reliable watermark embedding and extraction within QCGM-generated circuits. A clear upward trend is observed as the gate count increases, which can be attributed to the increased latent dimensionality. Specifically, the latent space width $f_w$ scales with gate number, and the overall latent size is given by $f_c \times f_h \times f_w$, matching the codeword length used for watermark embedding. Given a fixed message size (24 bits), longer codewords offer more redundancy for error correction, thereby improving detection reliability. In practice, increasing the default number of circuit gates further enhances robustness.

% \vspace{0.2cm}
\subsubsection{TPR under attacks}
We next evaluate the framework’s robustness against four types of attacks intended to erase the embedded watermark.

\noindent\textit{Gate Replacement.} 
Fig.~\ref{robustness}(b) presents the TPR after randomly replacing 1–5 gates in watermarked circuits. Across all gate counts and perturbation levels, the TPR remains above 85\%, and approaches 100\% for circuits containing 45–48 gates, confirming the method’s resilience to function-preserving substitutions.

\noindent\textit{Gate Appending.} Appending operations add 1–5 gates to the end of selected qubit lines. As shown in Fig.~\ref{robustness}(c), the TPR consistently exceeds 90\%, indicating robustness against appending-based obfuscation strategies that preserve output semantics but alter structure.

\noindent\textit{Gate Insertion.} In this scenario, 1–2 pairs of logically neutral gates (e.g., self-inverse pairs such as Hadamard–Hadamard) are inserted into random circuit positions. Detection performance remains strong across all settings, with TPR exceeding 95\%, as shown in Fig.~\ref{robustness}(d), demonstrating insensitivity to insertion-based perturbations.

\noindent\textit{Gate Deletion.} We randomly remove 1–3 gates from each watermarked circuit and evaluate the resulting TPR. As shown in Fig.~\ref{robustness}(e), the TPR remains above 90\% in all configurations, indicating that the watermark retains integrity even under structural deletions.

\begin{figure}[!t]
\centering
\includegraphics[width=1.0\linewidth]{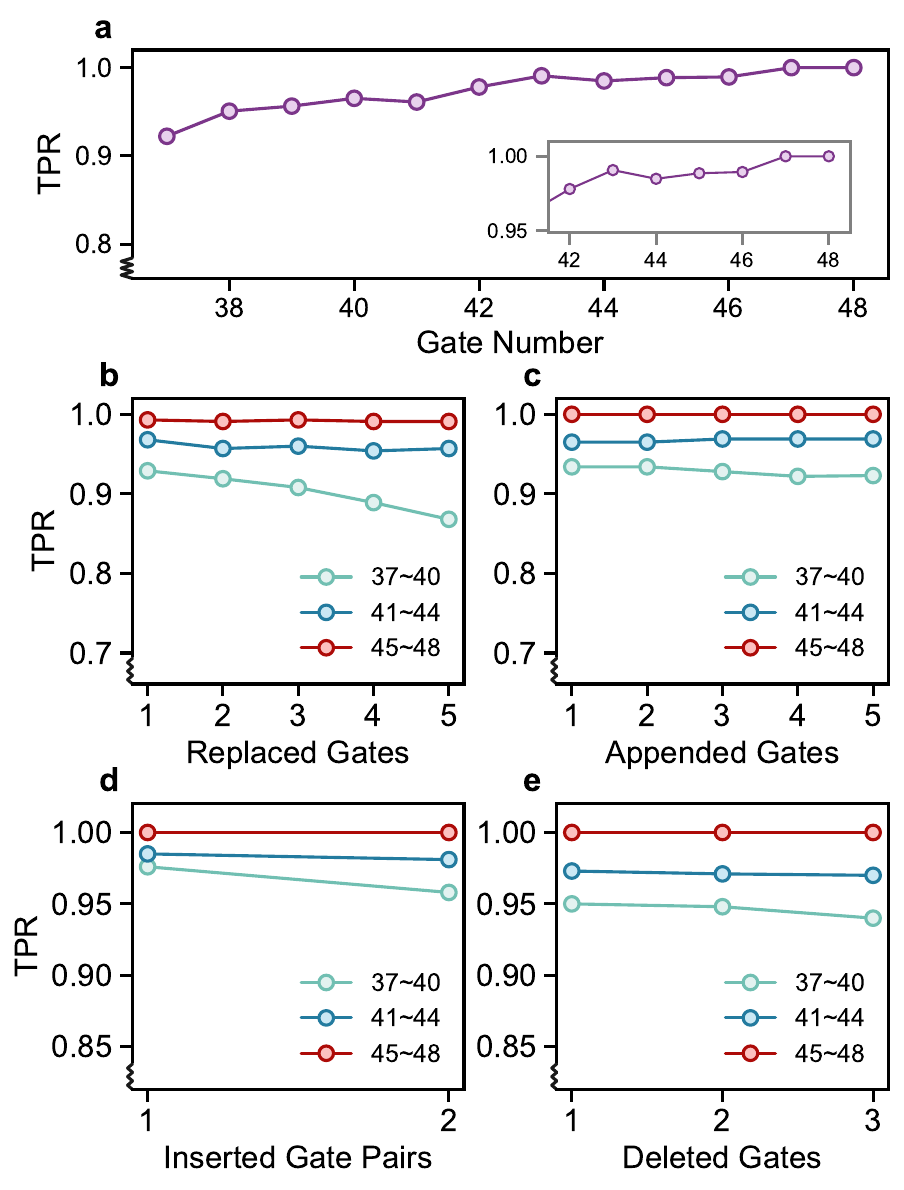}
\caption{True positive rate of watermark detection under different settings: (a) no-distortion, (b) replacement gates, (c) appending gates, (d) inserting gates and (e) deleting gates.}
\label{robustness}
\end{figure}

\begin{figure*}[!t]
\centering
\includegraphics[width=0.9\linewidth]{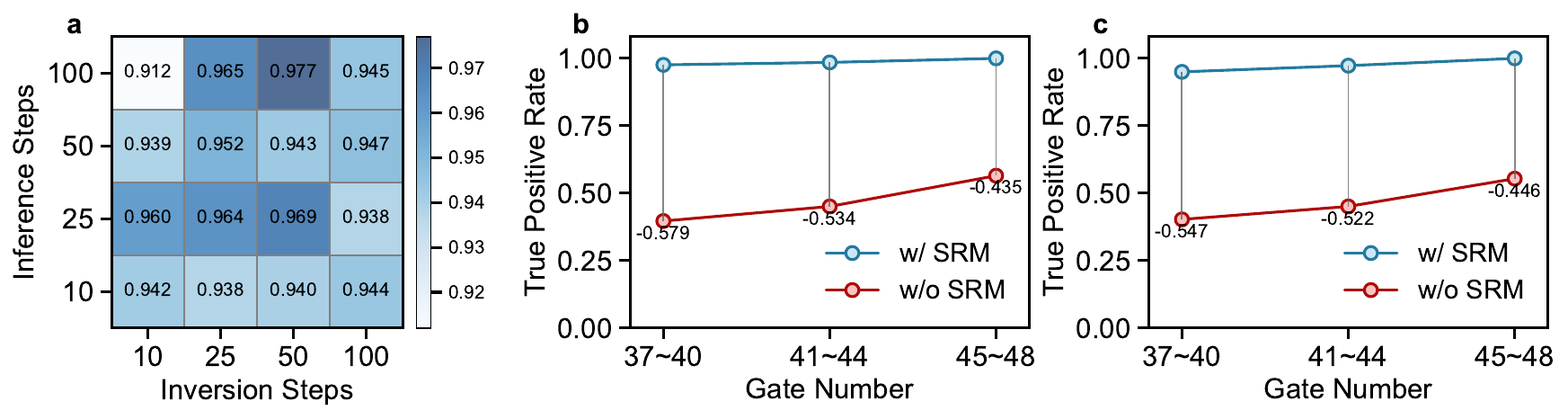}
\caption{Generalizability and ablation experiments: (a) the TPR of watermark detection with different inference and inversion steps; (b) the TPR of watermark detection with and without SRM under insertion attacks; (c) the TPR of watermark detection with and without SRM under deletion attacks.}
\label{fig_ablation_study}
\end{figure*}

\subsubsection{Comparative Studies}
To highlight the advantages of embedding watermarks during the circuit generation process, we compare our framework against representative watermarking approaches—specifically, the structural watermarking method of Roy et al.~\cite{Roy2025Watermarking}, which inserts redundant gate pairs into a generated circuit while preserving functionality. 
To ensure a fair and consistent comparison, all methods are evaluated under identical experimental conditions: circuits are generated with 8 qubits and the number of quantum gates ranging from 45 to 48.
Additionally, we also evaluate the performance under the four types of attacks described above: gate replacement, gate appending, gate insertion, and gate deletion.
The experimental results are presented in Table~\ref{Comparison}.
In the absence of attacks, all methods achieve 100\% accuracy in watermark extraction for copyright verification. 
However, under attack scenarios, the method of Roy et al~\cite{Roy2025Watermarking}. often fails to retain its embedded watermarks due to structural perturbations, resulting in extraction failure. 
In contrast, our method demonstrates strong robustness, reliably preserving watermark integrity and enabling successful extraction even in the presence of such attacks.

\begin{table*}[]
\centering
\caption{Comparison of the WQC}
\begin{tabular}{c|ccccccc}
\toprule[2pt]
\multicolumn{2}{c|}{Methods}        & Identity & Replacement & Appending & Insertion & Deletion &  \\ \midrule
\multicolumn{2}{c|}{WQC\cite{Roy2025Watermarking}-RGW}  & 1.000     & 0.398    & 0.402             & 0.405          & 0.201         &  \\
\multicolumn{2}{c|}{WQC\cite{Roy2025Watermarking}-RGIW} & 1.000     & 0.403    & 0.397             & 0.503          & 0.195         &  \\
\multicolumn{2}{c|}{Ours}           & 1.000    & 1.000    & 0.992             & 1.000          & 1.000         &  \\ \bottomrule[2pt]
\end{tabular}
\label{Comparison}
\end{table*}

% \vspace{.3cm}
\subsubsection{Generalizability}
\noindent \textit{Watermark Detection with Different Inference and Inversion Steps.}
In our experiments, the default configuration employs 50 generation steps during sampling and an identical 50-step inversion process for watermark extraction. However, real-world deployments may involve varied inference schedules depending on computational constraints or design preferences. To assess the generalization capability of the proposed framework under such conditions, we perform a series of cross-testing experiments by varying the number of generation steps (10, 25, 50, 100) and matching the inversion steps accordingly. The resulting TPR values across all tested combinations are reported in Fig.~\ref{fig_ablation_study}(a). In all cases, the framework maintains a TPR exceeding 91\%, demonstrating strong robustness and generalization across diverse generation–inversion step configurations.

% \subsection{Generalizability Test}
\vspace{0.2cm}
\noindent \textit{Robustness to Varying Guidance Scales.} In real-world scenarios, quantum circuit generation may be performed under different guidance scales, making it essential for the watermarking algorithm to remain effective across a range of such settings. To evaluate this, we generate 1,000 watermarked quantum circuits for each of four different guidance scales: 1.0, 2.5, 5.0, and 7.5. The corresponding watermark extraction results are summarized in Fig.~\ref{Adv_None_Kits}(a).

As the figure shows, the TPR remains consistently high ($>$95\%) across all tested scales, demonstrating that the proposed Q-Tag framework is robust to variations in guidance strength. This adaptability highlights the method's practical relevance for diverse generation configurations.

\begin{figure*}[ht]
\centering
\includegraphics[width=6.0in]{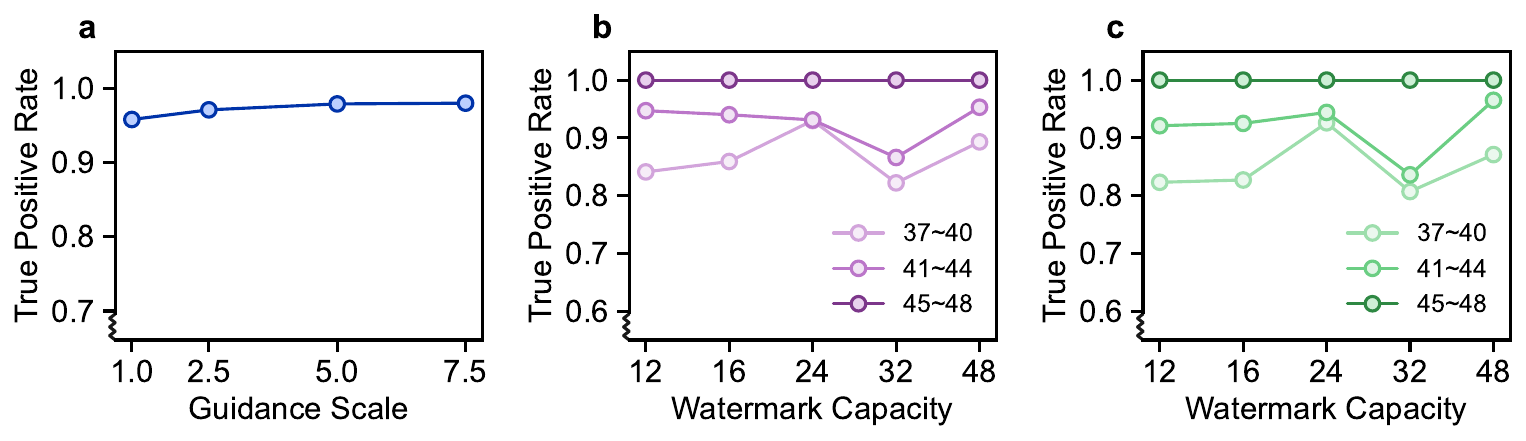}
\caption{(a) Comparison of the influence of different levels of text guidance on watermark extraction capability. (b) Different capacity sizes on watermark extraction capability without attacks. (c) Different capacity sizes on watermark extraction capability under attacks.}
\label{Adv_None_Kits}
\end{figure*}

\vspace{0.2cm}
\noindent \textit{Impact of Watermark Capacity.} We evaluate the TPR of watermark detection under varying watermark embedding capacities. Specifically, we test five capacity settings: 12, 16, 24, 32, and 48 bits. For each setting, we measure the TPR both without distortion and under distortion (including replacement and adversarial attacks), as shown in Fig. \ref{Adv_None_Kits}(b) and Fig. \ref{Adv_None_Kits}(c), respectively. To ensure consistency, the FPR is fixed at $10^{-3}$, and the corresponding detection thresholds are calibrated according to each bit length.

The results show that increasing the embedding capacity from 12 to 24 bits improves TPR, owing to a larger tolerance for bit errors under a fixed FPR—thereby enhancing detection reliability. However, when the capacity increases beyond 24 bits, TPR begins to decline. This is attributed to reduced redundancy in the error correction code, which weakens its ability to recover from distortion-induced bit errors. It is worth noting that, independent of embedding capacity, increasing the number of gates in the quantum circuit consistently improves TPR, due to the extended latent space available for robust watermark embedding.

% \vspace{.3cm}
\subsubsection{Importance of SRM}
To improve robustness against watermark attacks, we introduce a synchronization restoration mechanism. To evaluate its effectiveness, we conduct an ablation study. A random watermark $m$ is embedded into the model, and 10,000 quantum circuits (QCs) are generated. Each circuit is then subjected to insertion and deletion attacks, after which watermark extraction is performed both with and without SRM. The results are summarized in Fig.~\ref{fig_ablation_study}(b,c), where ``w/ SRM'' and ``w/o SRM'' denote extraction with and without synchronization restoration, respectively. The TPR achieved with SRM is consistently and significantly higher than that without SRM, with improvements exceeding 40 percentage points. These findings underscore the critical role of SRM in mitigating desynchronization effects introduced by structural perturbations, and highlight its importance in maintaining detection reliability under adversarial conditions.

\section{Conclusion}\label{Conclusion}
In conclusion, this work introduces the first watermarking framework explicitly designed for QCGMs, addressing a nontrivial challenge at the intersection of quantum design and model-level intellectual property protection. 
Our method departs from conventional circuit-level watermarking by embedding ownership signals directly into the generative process through a symmetric sampling mechanism, which non-invasively aligns watermark encoding with the model’s latent Gaussian prior—preserving fidelity without disrupting the generative distribution. 
To ensure robustness, we develop a synchronization restoration mechanism capable of recovering watermark integrity under structural perturbations. 
Beyond its immediate application to quantum circuits, the architectural principle underlying SRM, resynchronization of misaligned generative trajectories, offers a general insight for improving robustness in watermarking across structurally dynamic domains such as image editing, neural program synthesis, and video generation. 
This generalizability positions our framework as a foundational contribution to secure and accountable generation in both quantum and classical settings.

While this study establishes the first watermarking framework for quantum circuit generative models, several key challenges remain open. The current synchronization compensation mechanism relies on handcrafted heuristics, which may prove inadequate against sophisticated or adaptive adversaries. Furthermore, the watermark embedding capacity is limited, and the extraction process assumes a trusted environment. Future work could explore adaptive, learning-based alignment strategies, scalable embedding techniques with higher payload efficiency, and cryptographically verifiable extraction protocols. These directions would not only improve the resilience and scalability of quantum watermarking but also extend its applicability to a broader class of generative models. As quantum circuit generation increasingly underpins quantum software development, our framework lays the groundwork for secure, verifiable, and accountable content creation in the age of AI-assisted quantum design.

\bibliographystyle{unsrt}
\bibliography{references}

\vfill

\end{document}